\documentclass[conference]{IEEEtran} 
\usepackage{amsmath,amssymb,amsthm,mathtools}
\usepackage{bm,bbm}
\usepackage[dvipdfmx]{graphicx}
\usepackage{subcaption}
\usepackage{booktabs}
\usepackage{siunitx}
\usepackage{xcolor}
\usepackage{enumitem}
\usepackage[hidelinks]{hyperref}
\usepackage[capitalize,noabbrev]{cleveref}
\usepackage{microtype}
\usepackage{algorithm}
\usepackage{algorithmic}

\newtheorem{proposition}{Proposition}

\theoremstyle{definition}

\newtheorem{remark}{Remark}

\title{Mutual Information Estimation via Score-to-Fisher Bridge
for Nonlinear Gaussian Noise Channels}

\author{
\IEEEauthorblockN{Tadashi Wadayama}
\IEEEauthorblockA{Nagoya Institute of Technology\\
\texttt{wadayama@nitech.ac.jp}}
}

\begin{document}
\maketitle

\begin{abstract}
We present a numerical method to evaluate mutual information (MI) 
in nonlinear Gaussian noise channels by using denoising score matching (DSM) 
learning for estimating the score function of channel output.
Via de Bruijn's identity, 
Fisher information estimated from the learned score function yields 
accurate estimates of MI through a Fisher integral representation  
for a variety of priors and channel nonlinearities.
In this work, we propose a comprehensive theoretical foundation 
for the Score-to-Fisher bridge methodology, 
along with practical guidelines for its implementation. 
We also conduct extensive validation experiments, 
comparing our approach with closed-form solutions 
and a kernel density estimation baseline. 
The results of our numerical experiments demonstrate 
that the proposed method is both practical and efficient 
for MI estimation in nonlinear Gaussian noise channels.
Additionally, we discuss the theoretical connections between our score-based framework 
and thermodynamic concepts, 
such as partition function estimation and optimal transport.
\end{abstract}

\section{Introduction}
\label{sec:intro}

We consider additive Gaussian noise channels with a deterministic nonlinear front-end,
\begin{align}
Y_t = f(X) + Z_t,\qquad 
Z_t\sim\mathcal N(\bm 0,t\bm I_n),
\end{align}
and study how to estimate the mutual information (MI) $I(X;Y_t)$ efficiently and accurately.
A nonlinear front-end function can be represented by a deterministic function 
$f: \mathbb{R}^n \rightarrow \mathbb{R}^n$.
Such models arise widely due to nonlinear front-end saturation/clipping, 
DAC/ADC nonlinearity,
photonic modulators, optical fibers, or learned encoders implemented with 
neural networks. 
The value of the mutual information $I(X;Y_t)$ is important 
for evaluating the ultimate performance of coded communication systems;
it provides useful information for designing the system such as 
coding, preprocessing, or decoding/detection algorithms.
In a sensing system, the nonlinear 
front-end can be interpreted as a details of sensing process.
The value of the mutual information $I(X;Y_t)$ 
indicates how much information is transferred from the sensing 
target to the output.

When $f$ is nonlinear, evaluating $I(X;Y_t)$ is, in general, 
challenging in practice because accurate
posterior means $\mathbb E[X| Y_t]$ and normalized likelihoods 
are inaccessible or computationally intractable.
Let $W = f(X)$ be a transformed input.
The de~Bruijn identity immediately yields
\begin{align}
  \label{eq:bruijn-identity}
\frac{d}{dt} I(W;Y_t) = \frac{1}{2}J(Y_t)-\frac{n}{2t}
\end{align}
so that the mutual information $I(W;Y_t)$ can be obtained 
by integrating the Fisher information $J(Y_t)$.

Our key observation is that 
the Fisher information $J(Y_t)$ can be estimated from 
the score function of $Y_t$ which can be 
estimated by denoising score matching (DSM) learning from samples of $Y_t$,
and that the deterministic nonlinearity can be 
absorbed into a transformed input
$W=f(X)$, so that
\begin{align}
  \label{eq:mutual-info-equivalence}
  I(X;Y_t)=I(W;Y_t), \quad Y_t = W + Z_t,
\end{align}
holds for \emph{any} deterministic function $f$.  
This equality is the basis of our score-based estimator of MI.

By integrating (\ref{eq:bruijn-identity}) over $t\in[T,\infty)$, 
for general $f$, we have the {\em Fisher integral representation} of MI:
\begin{align}
  \label{eq:fisher-integral-representation}
  I(X;Y_T)=\frac12\int_T^\infty\!\Big(\frac{n}{t}-J(Y_t)\Big)\,dt
\end{align}
that depends only on the Fisher information of $Y_t$.

In this paper,
we propose a {score-based estimator of MI}\footnote{A software implementation to reproduce the numerical results in this paper is available
at \url{https://github.com/wadayama/score-to-fisher-mi}.}
tailored to the channel represented by $Y_t=f(X)+Z_t$.
The method learns the score of channel output, 
or the marginal density $p_{Y_t}$,
using  DSM learning and then obtain 
an estimate of the Fisher information 
$
\widehat J(Y_t).
$
From this, we obtain an estimate of the mutual information via an 
Fisher integral representation (\ref{eq:fisher-integral-representation}).
Crucially, the approach \emph{never} requires access to 
the posterior $p_{X|Y}(\cdot|\cdot)$.
All we need is possibility of forward channel simulation,
i.e., we assume that generation 
of samples $\bm x\sim p_X(\cdot)$, 
evaluations of $\bm w=f(\bm x)$, and Gaussian noise addition
are computationally feasible.

The main contributions of this paper are summarized as follows:
(1){Score-only MI estimation for nonlinear front-ends:}
we formalize a pipeline: DSM learning of $p_{Y_t}$ $\Rightarrow$ Fisher estimation
  $\Rightarrow$ Fisher integral representation—that applies to \emph{arbitrary} deterministic $f$
  as long as $p_{Y_t}$ can be easily sampled.
(2){Practical recipe:}
  we detail implementation choices that make the estimator stable and reproducible:
  logarithmic grids in $t$, tail corrections for the MI integral.
(3){Validation and use-cases:}
  we verify the estimator on Gaussian inputs where closed forms are available, and report
  experiments on composite nonlinear channels.

In linear Gaussian channels, multivariate de~Bruijn identity provides 
differential and integral representations of information quantities.
For nonlinear $f$, however, \emph{posterior-free}
and \emph{normalization-free} MI estimators remain limited. 
Building on DSM learning, our approach attaches directly 
to the marginal density $p_{Y_t}$,
circumventing posterior computations and partition functions.
Moreover, comparing with the conventional 
kernel density estimation (KDE) based estimators, 
our approach is robust to the high-dimensionality of the input space.
The estimate of MI can be used for many applications such as 
design of coding systems, optimization of preprocessing, 
sensor design, and system optimization, 
and design of detection and decoding algorithms.

In modern score-based modeling, the score function model serves as 
a tractable surrogate for otherwise intractable densities.
Our score-to-fisher bridge leverages this view to turn mutual 
information into a one-dimensional 
integral of score-based functionals along Gaussian smoothing, 
enabling accurate numerical MI estimation without 
high-dimensional density estimation.

\section{Preliminaries}
\label{sec:prelim}

\subsection{Notation}

Unless otherwise stated, all logarithms are natural (nats). Bold letters denote vectors in $\mathbb{R}^n$.
For a random vector $U$ with density $p_U$, we write $h(U)$ for its differential entropy, $I(\cdot;\cdot)$ for mutual information, and $\operatorname{Cov}(U)$ for its covariance matrix. 
The (scalar) Fisher information is defined via the score function 
$s_U(\bm u)\equiv\nabla_{\bm u}\log p_U(\bm u)$ as
\begin{align}
J(U)\equiv \mathbb{E}_{\bm u\sim p_U(\cdot)}\!\left[\|\nabla_{\bm u}\log p_U(\bm u)\|^2\right],
\end{align}
see, e.g.,~\cite{fisher1922,cover2006}.

\subsection{Related works}
\subsubsection{Estimation of Mutual Information} 

For channels with memory, simulation-based methods have been developed
to evaluate information rates without closed-form expressions.
Arnold \textit{et al.}~\cite{arnold2006TIT} proposed a simulation-based framework
that computes upper and lower bounds on the information rate by combining
auxiliary finite-state channel models with forward--backward recursions
and importance sampling; the method is broadly applicable 
to intersymbol interference (ISI) and
finite-state channels where exact evaluation is infeasible.
Building on this line, Sadeghi \textit{et al.}~\cite{sadeghi2009TIT}
optimized these bounds over the choice of auxiliary models, yielding
systematically improved estimates and a practical tool for design and
benchmarking of channels with memory.

Another line of work estimates the MI directly from its definition via
simulation-based plug-in estimators: one approximates $p_Y$ (or $h(Y)$)
using KDE~\cite{moon1995} or employs $k$-nearest neighbors
($k$NN) estimators~\cite{kraskov2004}. In the additive Gaussian setting,
recent theory shows that plug-in estimation of $h(X+Z)$ is in fact
optimal~\cite{goldfeld2019}.

\subsubsection{Theoretical Foundations and Thermodynamic Connections} 
The integral representation used in this work is deeply 
rooted in the interplay between information theory 
and thermodynamics, particularly through the heat equation. 
This connection dates back to Stam \cite{stam1959}, who derived the 
de Bruijn identity to prove the entropy power inequality, 
explicitly treating noise addition as a diffusion process. 
This line of inquiry was further developed by Brown \cite{brown1982} 
and Barron \cite{barron1986entropy}, who established the entropic central limit theorem, 
proving that the Fisher information and entropy monotonically 
converge to those of a Gaussian distribution during diffusion. 
In a modern context, Guo, Shamai, and Verdú \cite{guo2005} 
unified these concepts through the I-MMSE relationship, 
revealing that the derivative of mutual information is equivalent 
to the minimum mean-square error (and thus Fisher information) 
in Gaussian channels.
Our proposed Score-to-Fisher bridge can be viewed as a 
data-driven operationalization of these classical theoretical 
frameworks: while Stam and Barron utilized these relations for 
proofs of inequalities, we leverage them for the exact numerical 
estimation of information quantities via learned score functions.

\subsubsection{Diffusion Process and Fisher Information}
Recently, score-based representations have found applications beyond generative modeling, 
extending into information-theoretic analysis. 
For instance, Premkumar~\cite{premkumar2025neural} introduced ``neural entropy'' to quantify the information stored in diffusion models, relating the total entropy production along a diffusion path to a time integral of a \emph{Fisher-divergence} functional (i.e., an expected squared difference of score functions).
While their framework is motivated by a non-equilibrium thermodynamic view of diffusion generative models (and the resulting quantity depends on the chosen forward process), it shares a fundamental mathematical connection with our approach: both convert global information measures into one-dimensional integrals of score-based functionals, through diffusion/heat-equation identities (de~Bruijn's identity in our additive Gaussian-noise setting).

\subsection{AWGN channel and score function}

Let $X\in\mathbb{R}^n$ be a channel input random vector 
with known prior probability density function (PDF) $p_X$, 
and consider the AWGN channel family:
\begin{align}
  \label{eq:awgn-channel}
Y_t = X + Z_t,\quad Z_t \sim \mathcal N(\bm 0, t \bm I_n),\quad t>0,
\end{align}
where $t$ represents the noise variance.
The matrix $\bm I_n$ is the $n\times n$ identity matrix.
Let $p_{Y_t}$ denote the marginal density of $Y_t$. The score function of $p_{Y_t}$ is defined by
\begin{align}
  \label{eq:score-function}
  s_{Y_t}(\bm y) \equiv \nabla_{\bm y}\log p_{Y_t}(\bm y).
\end{align}

Throughout, we assume that sampling from $p_X$ is computationally feasible
by forward sampling or Langevin sampling using the score $\nabla_{\bm x}\log p_X(\bm x)$
(e.g., Rossky \textit{et al.} \cite{rossky1978}, Welling and Teh  \cite{welling2011sgld}).

The Fisher information of $Y_t$ is given by
\begin{align}
J(Y_t) \equiv \mathbb{E}_{\bm y\sim p_{Y_t}(\cdot)}\bigl[\|s_{Y_t}(\bm y)\|^2\bigr],
\end{align}
where $\|\cdot\|$ denotes the Euclidean norm.
Under standard regularity conditions, 
the following {\em Robbins-Tweedie identity} 
\cite{robbins1956, miyasawa1961, efron2011} holds for $t>0$:
\begin{align}
\mathbb{E}[X|Y_t=\bm y]=\bm y + t s_{Y_t}(\bm y).
\end{align}


\subsection{I--MMSE relation}

For the SNR-parameterized AWGN model
\begin{align}
Y_{\gamma}=\sqrt{\gamma} X + N,\quad N\sim\mathcal N(\bm 0,\bm I_n),\quad N\perp X,
\end{align}
the I--MMSE relation \cite{guo2005, palomar2006} states that the identity 
\begin{align}
\frac{d}{d\gamma} I(X;Y_\gamma) = \frac{1}{2}\mathrm{mmse}(\gamma),
\end{align}
holds for any snr $\gamma > 0$ where $\mathrm{mmse}(\gamma)$ is defined by 
\begin{align}
\mathrm{mmse}(\gamma) \equiv \mathbb{E}\bigl[\|X-\mathbb{E}[X| Y_\gamma]\|^2\bigr].
\end{align}
With the change of variables $\gamma=1/t$ (and noting $I(X;Y_\gamma)=I(X;Y_t)$),
an equivalent form for our variance parameterization $Y_t=X+Z_t$ is given by 
\begin{align}
\frac{d}{dt}\,I(X;Y_t) =-\frac{1}{2t^2}\mathrm{mmse}(t),
\end{align}
where 
$
\mathrm{mmse}(t)\equiv \mathbb{E}\bigl[\|X-\mathbb{E}[X| Y_t]\|^2\bigr].
$

\subsection{de Bruijn identity}

Let $h(Y_t)$ denote the differential entropy of $Y_t$. 
Consider the AWGN channel setup defined by (\ref{eq:awgn-channel}).
Then the (multivariate) de~Bruijn identity gives
\begin{align}
\frac{d}{dt} h(Y_t) =\frac{1}{2} J(Y_t).
\end{align}
The de Bruijn identity was discussed in \cite{stam1959, barron1986entropy},
which is proved by using the heat partial differential equation.
Intuitively, diffusion acts as a smoothing process that destroys information.
As illustrated in Fig.~\ref{fig:diffusion_process}, this process is characterized 
by the decay of Fisher information $J(Y_t)$ (which measures the `sharpness' or structural complexity of the distribution)
and the simultaneous increase in entropy $h(Y_t)$.
Our method computes the total information content by integrating this Fisher information profile.

\begin{figure}[t]
  \centering
  \includegraphics[width=\linewidth]{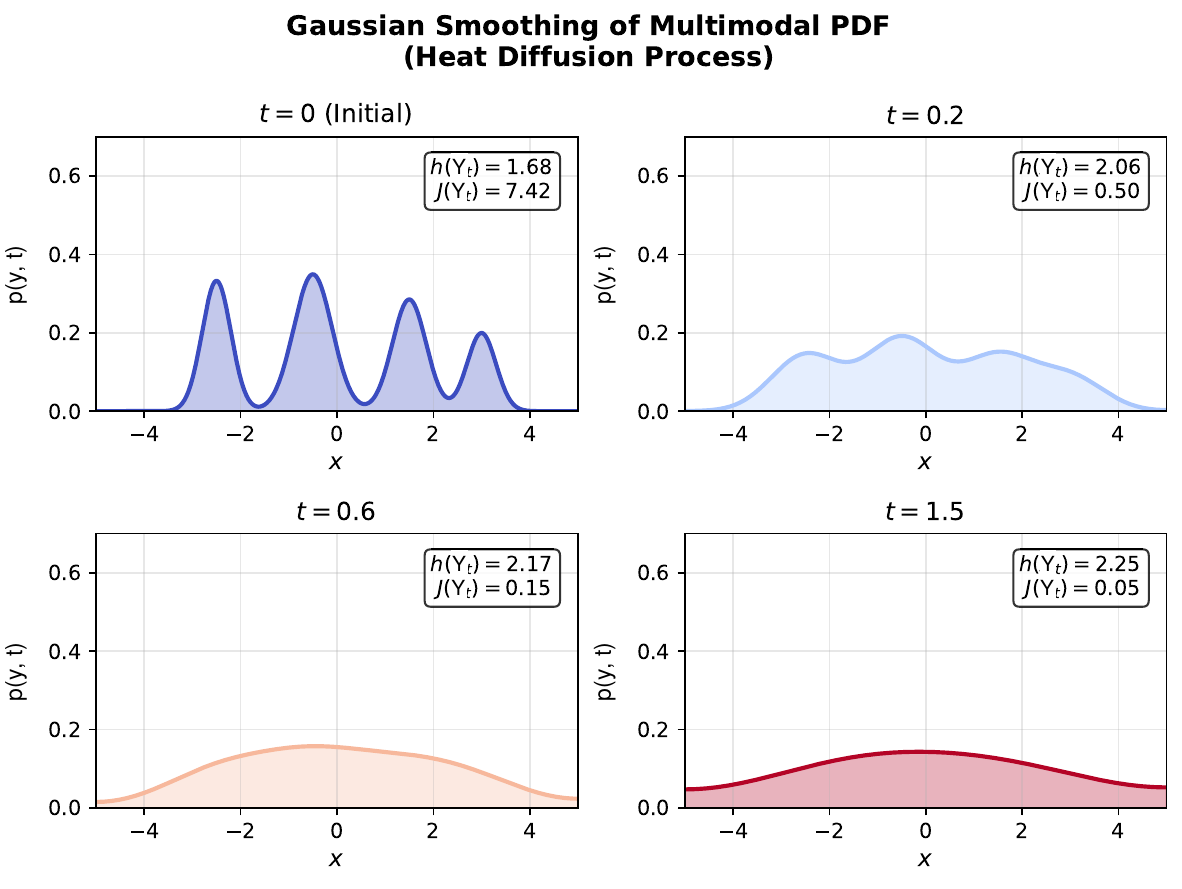}
  \caption{Visualization of the heat diffusion process for a 1D multimodal distribution (mixture of Gaussians). 
  As time $t$ progresses, the distribution $p(x,t)$ becomes smoother, leading to an increase 
  in differential entropy $h(X_t)$ and a rapid decay 
  in Fisher information $J(X_t)$. The Score-to-Fisher Bridge 
  method leverages this relationship via de Bruijn's identity: 
  the mutual information is essentially 
  accumulated Fisher-information gap relative to the Gaussian baseline 
  (the complexity of the distribution) over the diffusion time.}
  \label{fig:diffusion_process}
\end{figure}

\subsection{Denoising score matching (DSM) method}

DSM \cite{hyvarinen2005,vincent2011,song2019,song2021} provides a practical way to learn the score function 
of a Gaussian-smoothed distribution without computing any normalization constants or partition functions.
Let $p$ be a target density on $\mathbb{R}^n$ and define 
a Gaussian-smoothed sample:
\begin{align}
\bm v = \bm w + \sigma \bm \epsilon,\quad \bm w\sim p,
\quad \bm \epsilon\sim\mathcal N(\bm 0,\bm I_n).
\end{align}
Writing $p_\sigma \equiv p * \phi_{\sigma^2}$ for the convolution with a zero-mean Gaussian PDF $\phi_{\sigma^2}$ with
variance $\sigma^2$, the DSM loss function for a parametric score model $s_{\bm \theta}:\mathbb{R}^n\!\to\!\mathbb{R}^n$ ($\bm \theta$ represents the vector of adjustable parameters) is given by
\begin{align}
&\mathbb{E}_{\bm w,\bm \epsilon}\Big[\big\|s_\theta(\bm v)+\epsilon/\sigma\,\big\|^2\Big]
\\
&= \mathbb{E}_{\bm v\sim p_\sigma}\Big[\big\|\,s_\theta(\bm v)-\nabla_{\bm v}\log p_\sigma(\bm v)\,\big\|^2\Big] + \text{const},
\end{align}
where this equivalence was shown by Vincent \cite{vincent2011}. This means that 
minimizing the DSM loss function is equivalent to minimizing the expectation of the squared distance between the learned score function and the true score function,
which is often called the {\em Fisher divergence}.
This means that the unique minimizer of the DSM loss function satisfies 
$s_{\bm \theta}(\bm y) \simeq \nabla_{\bm y}\log p_\sigma(\bm y)$.
Equivalently, since $\bm v-\bm w=\sigma\bm \epsilon$, one may implement the loss function as $(\bm v-\bm w)/\sigma^2$:
\begin{align}
\mathbb{E}_{\bm w,\bm \epsilon}\big[\|s_\theta(\bm v)+(\bm v-\bm w)/\sigma^2\,\|^2\big],
\end{align}
which is an equivalent form of the DSM loss function.

In practice, we use a neural network to parameterize the score model 
$s_{\bm \theta}(\cdot)$. 
Such a score model is also called a {\em score network}. 
A mini-batch training with stochastic gradient descent 
or other variants such as Adam is typically used 
to minimize the DSM loss function.

\section{Estimating Mutual Information in AWGN Channels}

In this section, we present a method to estimate 
the MI for  AWGN channels described by
\begin{align}
Y_t = X + Z_t,\quad Z_t \sim \mathcal N(\bm 0, t\bm I_n),\quad t>0,
\end{align}
where $X \sim p_X$. The results obtained in this section will be 
extended to general nonlinear channels in Section \ref{sec:nonlinear-channel}.

\subsection{MI--Fisher relation}
For any Gaussian channel, the following proposition holds.
\begin{proposition}
For all $t>0$, the identity 
\begin{align} \label{eq:mi-fisher-relation}
  \frac{d}{dt}I(X;Y_t)=\frac{1}{2}J(Y_t)-\frac{n}{2t}
  \end{align}
  holds \cite{guo2005,palomar2006,reeves2018}. 
\end{proposition}
\begin{IEEEproof}
A sketch of proof is the following.
Using the definition of MI,
\begin{align}
I(X;Y_t)=h(Y_t)-h(Z_t)
\end{align}
with 
$h(Z_t)=(n/2)\log(2\pi e\,t)$, we have
\begin{align}
\frac{d}{dt}I(X;Y_t)=\frac{1}{2}J(Y_t)-\frac{n}{2t}
\end{align}
by de~Bruijn identity. 
\end{IEEEproof}
It can be remarked that, from the proposition, we immediately obtain 
the MMSE--Fisher relation
\begin{align}
  \label{eq:mmse-fisher-relation}
\mathrm{mmse}(t)=n t - t^2 J(Y_t)
\end{align}
by the I--MMSE relation.

Note that Reeves \textit{et al.}~\cite{reeves2018} investigated 
related identities in the
context of linear Gaussian channels, highlighting its implications for
non-Gaussian input distributions. Their work extends
the I--MMSE framework beyond the classical Gaussian setting.

\subsection{Fisher Integral representations of mutual information}

From the identity \eqref{eq:mi-fisher-relation}
and the fact that $\lim_{t\to\infty}I(X;Y_t)=0$ 
(since $Y_t$ and $X$ become independent in the limit),
integrating over $t\in[T,\infty)$ yields
\begin{align} \label{int_rep}
I(X;Y_T)
=& \frac{1}{2}\int_T^\infty\!\Big(\frac{n}{t}-J(Y_t)\Big)\,dt.
\end{align}
The integrand decays as $O(1/t^2)$, since
\begin{align}
0\ \le\ \frac{n}{t}-J(Y_t)\;=\;\frac{\mathrm{mmse}(t)}{t^2}
\ \le\ \frac{\operatorname{tr}\operatorname{Cov}(X)}{t^2}.
\end{align}
The integral representation directly follows by combining the 
de Bruijn identity \cite{stam1959} with the 
I–MMSE relation \cite{guo2005,palomar2006} as well.

\begin{remark}[Connection to GEXIT function]
  The integrand in our estimation formula corresponds to the channel GEXIT function $g(t)$. 
  Although typically defined via conditional entropy, it satisfies the following identities in the AWGN setting:
  \begin{align}
      g(t) \equiv \frac{d}{dt} h(X|Y_t) = -\frac{d}{dt} I(X; Y_t) = \frac{n}{2t} - \frac{1}{2}J(Y_t).
      \label{eq:gexit_fisher}
  \end{align}
  By the area theorem, the integral of $g(t)$ characterizes the mutual information, 
  thus linking our score-based Fisher integral to the GEXIT analysis commonly used in coding theory.
\end{remark}

\subsection{Estimation of Fisher Information}

Using DSM learning, we can estimate the score function of $Y_t$.
With $p \leftarrow p_X$ and $\sigma^2 \leftarrow t$, 
the learned score model can approximate the score of
$Y_t=X+Z_t$:
\begin{align}
s_{\bm \theta}(\bm y) \approx \nabla_{\bm y}\log p_{Y_t}(\bm y) = s_{Y_t}(\bm y),
\end{align}
where $s_{\bm \theta}(\cdot)$ can be implemented with a neural network 
with parameter vector $\bm \theta$.
This enables estimation of 
$J(Y_t)=\mathbb{E}_{\bm y \sim p_{Y_t}(\cdot)}\|s_{Y_t}(\bm y)\|^2$ 
from samples of $\bm x \sim p_X$ and Gaussian
noises.

In particular, for a fixed noise level $t$, the 
DSM loss function can be written as
\begin{align}
\mathcal{L}_t(\bm \theta)
\;\equiv\;
\mathbb{E}_{\bm x,\bm \epsilon}\!\left[\left\|\,s_{\bm \theta}(\bm x+\sqrt{t}\,
\bm \epsilon)+\bm \epsilon/\sqrt{t}\,\right\|^2\right] \\
=
\mathbb{E}_{\bm y}\!
\left[\left\|\,s_{\bm \theta}(\bm y)-s_{Y_t}(\bm y)\,\right\|^2\right] 
+ \text{const},
\end{align}
where the constant term does not depend on $\bm \theta$. Hence, with a
 sufficiently rich function class for $s_{\bm \theta}$
(e.g., a universal neural network) and sufficiently many samples, any sequence of empirical minimizers
achieves
$
\mathbb{E}_{\bm y}\!\left[\left\|\,s_{\bm \theta}(\bm y)-s_{Y_t}(\bm y)\,
\right\|^2\right] \rightarrow 0,
$
i.e., the learned score converges to the true score 
in the Fisher divergence sense.

Moreover, with i.i.d. draws 
\(\bm y_i=\bm x_i+\sqrt{t}\,\bm \epsilon_i\) from \(p_{Y_t} (i = 1,2,...,N)\),
the Monte Carlo estimator:
\begin{align}
\widehat J_t \equiv \frac{1}{N}\sum_{i=1}^N 
\|s_{\bm \theta}(\bm y_i)\|^2
\end{align}
satisfies \(\widehat J_t \xrightarrow{\text{a.s.}} J(Y_t)\) as \(N\to\infty\) (law of large numbers),
and \(\sqrt{N}\,(\widehat J_t - J(Y_t)) \Rightarrow \mathcal N\big(0,\operatorname{Var}(\|s_{Y_t}(\bm y)\|^2)\big)\) (central limit theorem),
so the estimation error decays at rate \(O(N^{-1/2})\).
In practice we plug in the learned score \(s_{\bm \theta}\) in place of \(s_{Y_t}\); under Fisher divergence
consistency of \(s_{\bm \theta}\), the plug-in estimator
\((1/N)\sum_{i=1}^N \|s_{\theta}(\bm y_i)\|^2\) 
is expected to be also consistent for \(J(Y_t)\).

\subsection{Gaussian input case for validation}

In this subsection we provide closed-form expressions for the Gaussian-input case
$X\sim\mathcal N(\bm 0, P\,\bm I_n)$ in the AWGN channel family 
$Y_t=X+Z_t$ with
$Z_t\sim\mathcal N(\bm 0, t\,\bm I_n)$ where $P > 0$ represents 
the signal power.
It is well known \cite{cover2006} that $I(X;Y_t)=(n/2)\log(1+P/t)$;
see also \cite{guo2005} for the relation to MMSE.

\begin{proposition}[Closed forms for Gaussian inputs]\label{prop:gauss-gexit}
Let $X\sim\mathcal N(\bm 0, P\,\bm I_n)$ and $Z_t\sim\mathcal N(\bm 0, t\,\bm I_n)$,
independent. Then, for all $t>0$, we have the following closed forms:
\begin{align}
 &\frac{1}{n}I(X;Y_t)=\frac{1}{2}\log\!\Big(1+\frac{P}{t}\Big), 
\end{align}  
\begin{align}
  \frac{1}{n}h(X\mid Y_t)=\frac{1}{2}\log\!\Big(2\pi e\,\frac{Pt}{P+t}\Big),
\end{align}  
\begin{align}  
  &\mathrm{mmse}(t)=n\,\frac{Pt}{P+t},\quad
  J(Y_t)=\frac{n}{P+t}.
\end{align}
\end{proposition}

\begin{IEEEproof}[Sketch of proof: ]
From the Gaussian model $Y_t=X+Z_t$, 
it is known that $I(X;Y_t)=(n/2)\log(1+P/t)$ 
holds.
Moreover $Y_t\sim\mathcal N(\bm 0,(P+t)\bm I_n)$, hence we have
\begin{align}
s_{Y_t}(\bm y)=\nabla_{\bm y}\log p_{Y_t}(\bm y)=-(P+t)^{-1}\bm y	
\end{align}
and
$J(Y_t)=\mathbb E\|s_{Y_t}(Y_t)\|^2=n/(P+t)$.
The remaining identities follow by substituting $J(Y_t)$ into
$\mathrm{mmse}(t)=n t - t^2 J(Y_t)$ and by
$h(X|Y_t)=h(Y_t)-I(X;Y_t)$ with 
$h(Y_t)=(n/2)\log\big((2\pi e)(P+t)\big).$
\end{IEEEproof}

\subsection{Linear Gaussian channel case}

Consider a linear Gaussian channel model:
\begin{align}
Y_t = \bm A X + Z_t,\qquad 
Z_t \sim \mathcal N(\bm 0, t \bm I_n),
\label{eq:linear-channel}
\end{align}
where $\bm{A} \in \mathbb R^{n\times n}$ is a square, invertible matrix and
$X \sim \mathcal N(\bm 0, P \bm I_n)$. 
This setting preserves the Gaussian structure while introducing
nontrivial correlation through the linear mixing due to $A$.

\begin{proposition}[Closed forms for linear Gaussian channels]
Let $X \sim \mathcal N(\bm 0, P \bm I_n)$, 
$\bm{A} \in \mathbb R^{n\times n}$ nonsingular, and 
$Z_t \sim \mathcal N(\bm 0, t \bm I_n)$, independent. 
Then for all $t>0$:
\begin{align}
I(X;Y_t) 
&= \tfrac{1}{2}\log \det\!\Bigl(I_n + \tfrac{P}{t} \bm{A} \bm{A}^\top\Bigr),
\label{eq:I-linear} \\
\mathrm{mmse}(t) 
&= \mathrm{tr}\!\left(\Bigl(P^{-1} I_n + \tfrac{1}{t} \bm{A}^\top \bm{A} \Bigr)^{-1}\right),
\label{eq:mmse-linear} \\
J(Y_t) 
&= \mathrm{tr}\!\Bigl((P \bm{A} \bm{A}^\top + t \bm{I}_n)^{-1}\Bigr).
\label{eq:fisher-linear} 
\end{align}
\end{proposition}

\begin{IEEEproof}[Sketch of proof:]
Since $AX$ is Gaussian with covariance $P \bm{A} \bm{A}^\top$, 
we have $Y_t \sim \mathcal N(\bm 0, P \bm{A} \bm{A}^\top + t \bm{I}_n)$. 
The mutual information \eqref{eq:I-linear} follows from the
determinant formula for Gaussian differential entropy 
\cite{cover2006,palomar2006}. 
The Fisher information \eqref{eq:fisher-linear} is obtained directly from
the Gaussian score function 
$s_{Y_t}(\bm y) = -(P \bm{A} \bm{A}^\top + t \bm{I}_n)^{-1}\bm y$. 
The MMSE formula \eqref{eq:mmse-linear} comes from 
the linear MMSE estimator of $X$ given $Y_t$. 
\end{IEEEproof}

\section{Estimation of Mutual Information}

Using the score-based estimation of the Fisher information, 
we can evaluate $I(X;Y_T)$ by numerical integration.
In this section, we present details of the method 
to estimate the MI for the AWGN channel family described by (\ref{eq:awgn-channel}).

\subsection{Score-based Estimation Method for MI}

The proposed estimation algorithm is summarized in Algorithm \ref{MI_MMSE_alg}.
The core of the algorithm is the DSM learning for $p_{Y_t}$ 
and the estimation of the Fisher information from the score model.

\begin{algorithm}[htbp]
  \caption{Score-to-Fisher bridge: Estimation of MI}
  \label{MI_MMSE_alg}
 \begin{algorithmic}[1]
   \STATE Prepare a logarithmically spaced grid over $t$ (details of the log grid are described in Subsection \ref{sec:log-domain-trapezoid-integration}).
   \STATE Sample generation: Generate a random mini-batch $\mathcal B$ consisting of 
   multiple pairs of $(\bm x, \bm \epsilon)$ for mini-batch training with 
   forward sampling of $X$:
   $\bm x \sim p_X, \ 
     \bm\epsilon \sim \mathcal N(\bm 0, \bm I_n).
   $
   \STATE DSM learning for $p_{Y_t}$: train the score model parameters by minimizing the empirical DSM loss:
   \begin{align}
     \sum_{(\bm x, \bm \epsilon) \in \mathcal B}\big\|\,s_{\bm \theta}(\bm x+\sqrt{t}\,\bm\epsilon)\;+\;\bm\epsilon/\sqrt{t}\,\big\|^2.
   \end{align}
   Any stochastic optimizer,  
     such as Adam, can be applied to minimize the empirical loss function. 
     \STATE Estimating the Fisher information: for each $t$, generate $\bm y_1,\ldots,\bm y_N$ (e.g., by forward sampling $\bm y_i=\bm x_i+\sqrt{t}\,\bm\epsilon_i$), and compute the estimated Fisher information:
     \begin{align}
       \widehat J_t = \frac{1}{N}\sum_{i=1}^N \big\| s_{\bm \theta}(\bm y_i) 
       \big\|^2
     \end{align}
     as an approximation of $J(Y_t)$. 
     \STATE Estimating the MI via the Fisher integral representation:
     \begin{align}
      \label{eq:mi-estimation}
       \widehat{I(X;Y_T)} = \frac{1}{2}\int_T^{\infty}\Big(\frac{n}{t}-\widehat J_t\Big)\,dt,
     \end{align}
     evaluated numerically.
\end{algorithmic}
\end{algorithm}

It should be noted that if the probability density $p_{Y_t}$ of the channel output $Y_t$ is known, and its score function $s_{Y_t}(y) = \nabla_y \log p_{Y_t}(y)$ is available in closed form (e.g., in the case of an identity function $f$ with a Gaussian Mixture Model (GMM) input $X$), 
then Step 3 (DSM learning) becomes unnecessary. 
In such cases, one can directly use this analytical score function in Step 4 to perform the Monte Carlo estimation of the Fisher information. 
This approach eliminates the score approximation error while retaining the advantages over conventional methods in terms of computational cost and high-dimensional scalability. 
However, the primary focus of this paper is on more general nonlinear channels where such an analytical solution is intractable or unavailable.

in Algorithm \ref{MI_MMSE_alg}, we consider two alternatives for the DSM learning:
\begin{itemize}
  \item {Scheme A (per-$t$ training).}
  For each fixed $t$, minimize
  \[
    \mathbb{E} \big\| s_\theta(\bm x+\sqrt{t}\bm\epsilon) + \bm\epsilon/\sqrt{t} \big\|^2
  \]
  to obtain the learned score model. 
  \item {Scheme B (noise-conditional model).}
  Learn $s_\theta(\bm y,t)$ by minimizing
  \begin{align} 
    \mathbb{E}_{t\sim\mathcal T}\,
    \lambda(t)\ \mathbb{E}_{\bm x,\bm\epsilon}\, \big\| s_{\bm \theta}(\bm x+\sqrt{t}\,\bm\epsilon,t) + \bm\epsilon/\sqrt{t} \big\|^2,
  \end{align}
  where $\mathcal T$ is a training distribution over $t$ (e.g., logarithmically spaced as described later) 
  and $\lambda(t)$ is a weight function, e.g., $\lambda(t)=t$.
\end{itemize}
In the case of Scheme B, the score network has two inputs: 
the channel output $\bm y$ and the noise level $t$.

\subsection{Log-domain trapezoid integration}
\label{sec:log-domain-trapezoid-integration}
The integrand in the Fisher integral representation 
(\ref{eq:fisher-integral-representation}) of the MI is
smooth and may have a long tail. The choice of the integration grid is
critical for achieving high numerical accuracy. In this subsection,
we present a geometric (log-spaced) grid for the integration.
For $t\in[t_{\min},t_{\max}]$ with $M$ points, take $\log t$ to be equally spaced:
\begin{align}
  t_k = t_{\min}\left(\frac{t_{\max}}{t_{\min}}\right)^{\frac{k}{M-1}},
  \qquad k=0,1,\ldots,M-1.
  \label{eq:geom-grid}
\end{align}
Equivalently in standard deviation, $\sigma_k \equiv \sqrt{t_k}$ with
$\sigma_k=\sigma_{\min}(\sigma_{\max}/\sigma_{\min})^{k/(M-1)}$.

Rather than applying the trapezoid rule directly on the non-uniform $t$-grid,
we perform the change of variables $u=\log t$ to obtain uniform spacing 
in the log domain.
Using the integral representation (\ref{eq:fisher-integral-representation}),
and substituting $u=\log t$ (so $dt=e^u du$ and $t=e^u$), we obtain
\begin{align}
  I(X;Y_T)
  = \frac12\int_{\log T}^{\infty}\!\Big(n - e^u\,J(Y_{e^u})\Big)\,du.
  \label{eq:log-integral}
\end{align}

Since the $t$-grid \eqref{eq:geom-grid} corresponds to uniform spacing
in $u=\log t$ with step size $\Delta_u \equiv \log(t_{\max}/t_{\min})/(M-1)$,
we apply the standard trapezoid rule with uniform spacing:
\begin{align}
  \widehat{I}(X;Y_T)
  \approx \frac12 \sum_{k=k_T}^{M-2} \frac{\ell(u_k)+\ell(u_{k+1})}{2}\,\Delta_u
  \;+\; \text{tail},
\end{align}
where $\ell(u)\equiv n - e^u\,\widehat{J}(Y_{e^u})$,
and $k_T$ is the smallest index with $t_{k_T}\ge T$.
This log-domain integration offers several benefits over direct $t$-domain
trapezoid integration: (1) uniform grid spacing reduces discretization error,
(2) the transformation naturally handles a wide dynamic range in $t$, and
(3) convexity bias is significantly reduced. In our experiments,
this approach improved MI estimation accuracy by over an order of magnitude.

The tail contribution above $t_{\max}$ can be estimated using the asymptotic behavior.
For Gaussian inputs, $\mathrm{mmse}(t)\to \operatorname{tr}\mathrm{Cov}(X)$ as $t\to\infty$, giving
\begin{align}
  \text{tail}
  \approx \frac{1}{2}\int_{t_{\max}}^\infty \frac{\mathrm{mmse}(t)}{t^2}\,dt
  \approx \frac{1}{2}\,\frac{\operatorname{tr}\mathrm{Cov}(X)}{t_{\max}}.
  \label{eq:tail-correction}
\end{align}
The trace $\operatorname{tr}\mathrm{Cov}(X)$ can be estimated from samples of $p_X$ and used as a tail correction.

While the tail correction formula \eqref{eq:tail-correction} is exact for Gaussian inputs,
its leading-order term $\mathrm{tr}\,\mathrm{Cov}(X)/t_{\max}$ 
remains valid for arbitrary inputs, since $\mathrm{mmse}(t)\to \mathrm{tr}\,\mathrm{Cov}(X)$ as $t\to\infty$.
Therefore, the same tail correction can be applied as a universal approximation,
with the understanding that higher-order terms may depend on the non-Gaussianity of $X$.

\subsection{Stein Calibration}
In practice, the learned score function $s_{\theta}(\mathbf{y})$ may exhibit a scaling bias. 
To mitigate this, we apply Stein calibration to the score estimates. Leveraging Stein's identity, 
we compute a scalar correction factor $c$ such that the calibrated score 
$\tilde{s}_{\theta}(\mathbf{y}) = c s_{\theta}(\mathbf{y})$ satisfies the empirical 
moment constraint $\mathbb{E}[\mathbf{Y}_t^{\top} \tilde{s}_{\theta}(\mathbf{Y}_t)] = -n$, 
assuming centered data. The estimated Fisher information is then computed using 
this calibrated score.

\subsection{Alternative Score Estimation via Velocity Fields}
While we employed DSM to estimate the score function, 
the SFB framework is inherently estimator-agnostic. Recent generative paradigms, 
such as {\em flow matching} \cite{lipman2023flowmatching} and {\em rectified flow} \cite{liu2023rectifiedflow}, 
learn a velocity field 
that drives the probability flow, typically along straight-line trajectories. 
Since the score function can be mathematically derived from the learned velocity field 
(given the specific scaling factor), these methods offer a promising alternative 
route for score estimation. Leveraging such techniques could potentially improve 
the numerical stability of MI estimation, particularly in regimes where DSM struggles 
with curvature in the diffusion path.
See Appendix G for more details.

\section{Nonlinear channel}
\label{sec:nonlinear-channel}

Let us consider our primal target channel:
$$
Y_t = f(X) + Z_t,
$$
where $f: \mathbb{R}^n \rightarrow \mathbb{R}^n$ 
is a deterministic function\footnote{An extension to a general case 
$f: \mathbb{R}^n \rightarrow \mathbb{R}^m$ is straightforward but 
we assume $m=n$ for simplicity.}. 
The mutual information equivalence 
presented in the following proposition is a key observation 
which enables the score-based estimation to be applied 
to general nonlinear channels.

\subsection{Mutual information equivalence}

Let $W \equiv f(X)$ be a measurable deterministic mapping, and consider $Y_t=W+Z_t$ with $Z_t\sim\mathcal N(\bm 0,t\bm I_n)$ independent of $(X,W)$. Then the conditional law of $Y_t$ given $X$ depends on $X$ only through $W$.
   This means that $X\to W\to Y_t$ forms a Markov chain and 
   equivalently $Y_t\perp X\mid W$.

We can prove {\em mutual information equivalence} in the following proposition.
\begin{proposition}
\label{prop:mutual-info-equivalence}
Let $Y_t = W + Z_t (t > 0)$ with $Z_t\sim\mathcal N(\bm 0,t\bm I_n)$ independent of $(X,W)$ and $W = f(X)$,
where $f$ is a measurable deterministic function.
Then we have 
\begin{align}
I(X;Y_t)=I(W;Y_t).
\end{align}
\end{proposition}
\begin{IEEEproof}
By the chain rule of mutual information, we have
\begin{align}
I(X;Y_t)=I(W;Y_t)+I(X;Y_t\mid W)-I(W;Y_t\mid X).
\end{align}
Here $I(X;Y_t\mid W)=0$ because of the conditional independence $Y_t\perp X\mid W$, and $I(W;Y_t\mid X)=0$  holds because $W$ is a function of $X$.
Hence $I(X;Y_t)=I(W;Y_t)$.
An equivalent proof uses the data processing inequality twice: $X\to W\to Y_t$ gives $I(X;Y_t)\le I(W;Y_t)$, while $W\to X\to Y_t$ (since $W=f(X)$) gives $I(W;Y_t)\le I(X;Y_t)$, yielding equality.
\end{IEEEproof}

From Proposition \ref{prop:mutual-info-equivalence}, we have 
the MI equivalence $I(X;Y_t)=I(W;Y_t)$, which
implies that the score-based estimation method presented in the previous section can be applied to evaluate $I(X;Y_T)$.
Namely, instead of $X$, from the samples of $W$, 
we can estimate the Fisher information of $Y_t$ 
and  the estimate of the Fisher information can be used 
	to evaluate $I(W;Y_t)$. The equivalence relation $I(X;Y_t)=I(W;Y_t)$ 
  enables to obtain the estimate of $I(X;Y_t)$ from the estimate of $I(W;Y_t)$.
Note that no invertibility of $f$ is required in 
this estimation method.

\subsection{Score-based estimation method for nonlinear channels}

In the following discussion, we assume that the evaluation of $f$ is 
computationally feasible. Namely, sampling from $p_W$ is feasible.
Some parts of Algorithm \ref{MI_MMSE_alg} should be
replaced by the following steps:
\begin{itemize}
  \item Sample $W$ by $W = f(X)$ where $X \sim p_X$.
  \item Estimate the Fisher information of $Y_t$ from the samples of $W$ 
  by the score-based estimation in Algorithm \ref{MI_MMSE_alg}.
  \item Use the estimated Fisher information 
  to evaluate the Fisher integral form of MI.
  \item Use the equality $I(X;Y_t)=I(W;Y_t)$ to obtain the estimate of $I(X;Y_t)$.
\end{itemize}
The estimation procedure is still a single-loop procedure, 
which does not require the access to the posterior density $p_{X|Y_t}(\cdot|\cdot)$.

\subsection{Error Analysis and Convergence Rates}
This section provides a qualitative analysis of the estimation error for the proposed estimator $\hat{I}$. 
The total estimation error can be decomposed into four distinct components:
\[
    |\hat{I} - I| \;\le\; 
    \varepsilon_{\mathrm{score}} + \varepsilon_{\mathrm{MC}} + \varepsilon_{\mathrm{quad}} + \varepsilon_{\mathrm{tail}}.
\]
Each component behaves as follows under standard regularity assumptions:

\begin{enumerate}
    \item \textbf{Score Approximation Error ($\varepsilon_{\mathrm{score}}$):} 
    This term arises from the imperfect learning of the score function $s_\theta \approx \nabla \log p_t$. It scales with the square root of the Fisher divergence loss. Improving the neural network capacity and training budget reduces this bias.
    
    \item \textbf{Monte Carlo Variance ($\varepsilon_{\mathrm{MC}}$):} 
    Since the Fisher information $J(Y_t)$ is estimated via sample averaging, this stochastic error follows the standard Central Limit Theorem scaling of $O(N^{-1/2})$, where $N$ is the batch size per grid point.
    
    \item \textbf{Quadrature Error ($\varepsilon_{\mathrm{quad}}$):} 
    The discretization of the integral over $u=\log t$ using the trapezoidal rule introduces a deterministic error. On a uniform grid with $M$ points (step size $\Delta u \propto M^{-1}$), this error decays rapidly as $O(M^{-2})$.
    
    \item \textbf{Truncation Error ($\varepsilon_{\mathrm{tail}}$):} 
    Restricting the integration to $[t_{\min}, t_{\max}]$ introduces bias from the missing tails. The error typically decays as $O(t_{\max}^{-1})$ for the large-$t$ tail and polynomially (e.g., $O(t_{\min}^\beta)$) for the small-$t$ tail.
\end{enumerate}

To minimize the total error efficiently, one must balance these components.
Since the quadrature error $\varepsilon_{\mathrm{quad}} \sim O(M^{-2})$ decays much faster than the Monte Carlo error $\varepsilon_{\mathrm{MC}} \sim O(N^{-1/2})$, a relatively coarse grid $M$ is often sufficient compared to the sample size $N$. 
Theoretically, balancing the rates suggests a scaling of $M \asymp N^{1/5}$.
In practice, we recommend fixing a sufficiently wide integration range $[t_{\min}, t_{\max}]$ to suppress $\varepsilon_{\mathrm{tail}}$, and then prioritizing a large batch size $N$ to mitigate the dominant Monte Carlo variance.

\section{Numerical experiments}

\subsection{Gaussian input case for validation}
To validate the proposed score-based estimation method, we first consider the
Gaussian input case where exact closed-form expressions are available. 
Specifically, we consider the Gaussian input case $X \sim
\mathcal{N}(\bm 0, P \bm I_n)$ with the AWGN channel model 
$Y_t = X + \sqrt{t} Z$, where $Z \sim \mathcal{N}(\bm 0, \bm I_n)$.

\subsubsection{Experimental Setup}

We implemented the DSM learning for $p_{Y_t}$ with the following configuration:
  The score network is implemented as a fully connected neural network 
  consisting of 3 hidden layers,
  each containing 128 units with SiLU (Swish) activation functions. 
  The input and output dimensions match
   the problem dimension $n$, where we test with $n \in \{4, 8, 16\}$ 
   to evaluate scalability.
  For the training configuration, we set the signal power to $P = 1.0$ 
  and construct a logarithmically
  spaced grid of $M = 10$ noise variance values ranging 
  from $t \in [P/200, 200P]$, which covers four
  orders of magnitude. At each noise level $t$, we 
  train a separate score network for 300 iterations
  using mini-batches of 4096 samples. The optimization 
  is performed using the Adam optimizer with a
  learning rate of $10^{-3}$, and gradient clipping is 
  applied with a maximum norm of 1.0 to ensure
  stable training.

  \subsubsection{Validating the Fisher Information Estimate}

  For Fisher information estimation, we generate 100,000 Monte Carlo samples at each noise variance $t$
  to ensure statistical reliability. The ground truth values are computed using the analytical formula
  $J(Y_t) = n/(P + t)$, which is available for Gaussian input distributions.

Figure~\ref{fig:gaussian_fisher} presents the estimated Fisher information $\hat{J}(Y_t)$ compared with
 the analytical ground truth $J(Y_t) = n/(P+t)$ for dimensions $n \in \{4, 8, 16\}$. The results
demonstrate excellent agreement between the DSM-based estimates and theoretical values across multiple
orders of magnitude in the noise variance $t$.
\begin{figure}[htbp]
    \centering
    \includegraphics[width=0.48\textwidth]{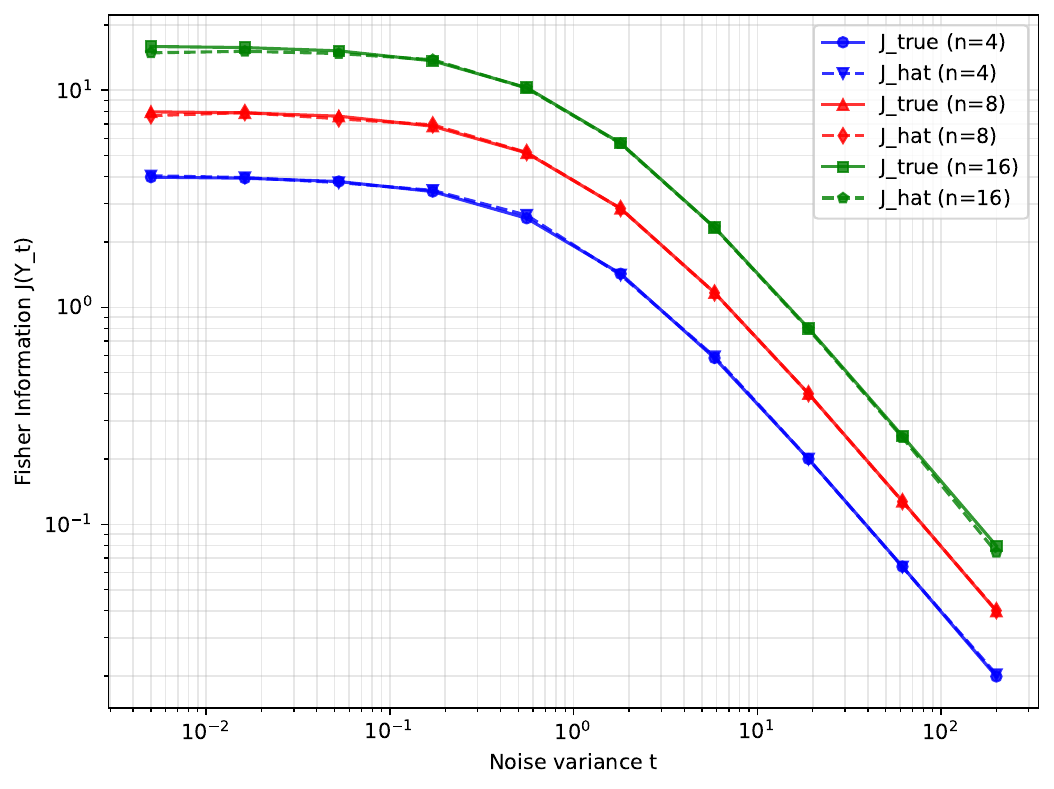}
    \caption{Fisher information $J(Y_t)$ versus noise variance $t$ for Gaussian input with $P=1$ and
dimensions $n \in \{4, 8, 16\}$. Solid lines with circular markers show ground truth values $J(Y_t) =
 n/(P+t)$, while dashed lines with triangular markers show estimated values $\hat{J}(Y_t)$. }
    \label{fig:gaussian_fisher}
\end{figure}
The results confirm that:
\begin{enumerate}
    \item The DSM-based score estimation accurately captures the Fisher information across a wide range
 of noise levels ($t \in [0.005, 200]$).
    \item The method scales well with dimension, maintaining low relative errors even for $n=16$.
    \item The median relative error remains below 1.3\% for all tested dimensions, validating the
theoretical framework.
\end{enumerate}

  \subsubsection{Mutual Information Estimation}

  We demonstrate the complete pipeline from 
  the Score-to-Fisher bridge for MI estimation using the integral
  representation of MI.
  The experiment uses the same configuration as previous tests with $n=4$ and $P=1.0$, implementing
  log-domain trapezoid integration with tail correction.
  Figure~\ref{fig:mi_estimation} presents the estimated mutual information (per symbol) 
  $\hat{I}(X;Y_t)/n$ compared with
  the analytical ground truth $I(X;Y_t)/n = (1/2)\log(1 + P/t)$ for Gaussian input. The DSM-based
 estimates demonstrate excellent agreement with the theoretical values across the entire noise variance
 range.

  \begin{figure}[htbp]
      \centering
      \includegraphics[width=0.48\textwidth]{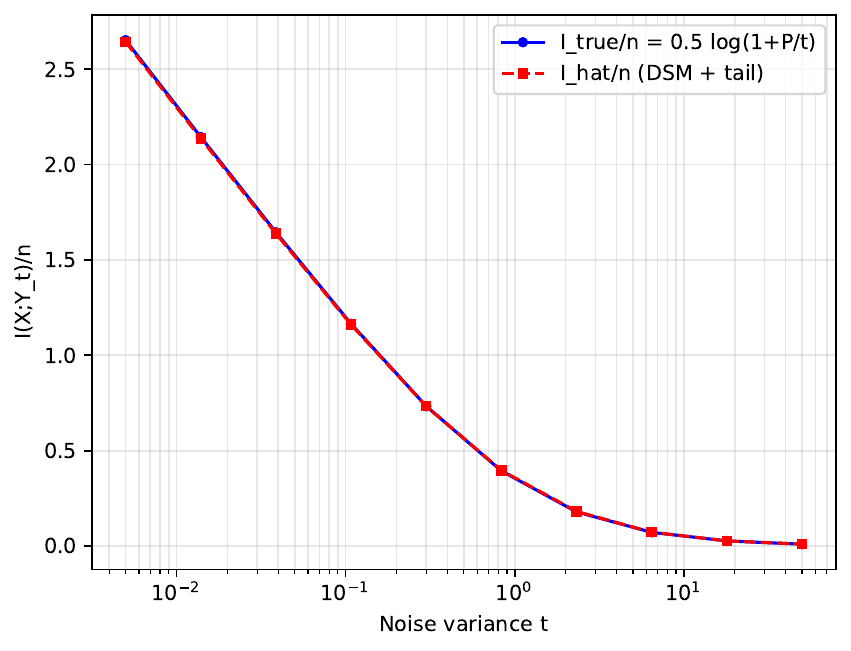}
      \caption{Mutual information (per symbol) $I(X;Y_t)/n$ versus noise variance $t$ for Gaussian input. }
      \label{fig:mi_estimation}
  \end{figure}

  A key contribution of this experiment is the log-domain trapezoid integration method. Compared to standard
  trapezoid integration on the $t$-grid, the log-domain approach improved estimation accuracy by over an
  order of magnitude, reducing median relative errors from approximately 26\% to 0.4\%. This substantial
  improvement arises from: (1) uniform grid spacing in the log domain, (2) reduced convexity bias, and
  (3) better numerical handling of the wide dynamic range in $t$ values.

  The tail correction, based on the asymptotic behavior $\text{mmse}(t) \sim \text{tr Cov}(X)$ for large
  $t$, proves essential for maintaining accuracy across the full integration range. Without tail
  correction, the 90th percentile error degrades significantly to over 40\%.

These validation results on Gaussian inputs, where analytical solutions exist, provide strong
evidence for the reliability of the proposed Score-to-Fisher bridge method and related theoretical framework presented in the previous sections.

\subsubsection{Noise-conditional model}

  We investigate a noise-conditional score model $s_{\bm \theta}(\bm y, t)$ (Scheme B) 
  that takes both
  observation $\bm y$ and noise variance $t$ as inputs, eliminating the need for separate
  per-$t$ training. The model uses Gaussian Fourier projection to embed scalar $t$
  values, which are injected into each hidden layer. Training samples $t$ from a
  log-uniform distribution over $[P t_{\min}, P t_{\max}]$ with loss
  weighting $\lambda(t) = t$ for 20,000 steps.

  For Gaussian input with $n=4$ and $P=1.0$, this approach achieves 6.61\% median
  relative error across 12 evaluation points spanning $t \in [0.005, 50]$.
  Figure~\ref{fig:noise_conditional_mi} demonstrates reasonable agreement with analytical
  values across four orders of magnitude in noise variance. We can observe the tendency that 
  the noise-conditional DSM estimate slightly underestimates the MI 
  at low noise variance. This result may be due to insufficient training iterations 
  for a conditional model.
  However, this unified framework requires only one neural network instead of 12 separate models,
  significantly reducing computational cost.

  \begin{figure}[htbp]
      \centering
      \includegraphics[width=0.48\textwidth]{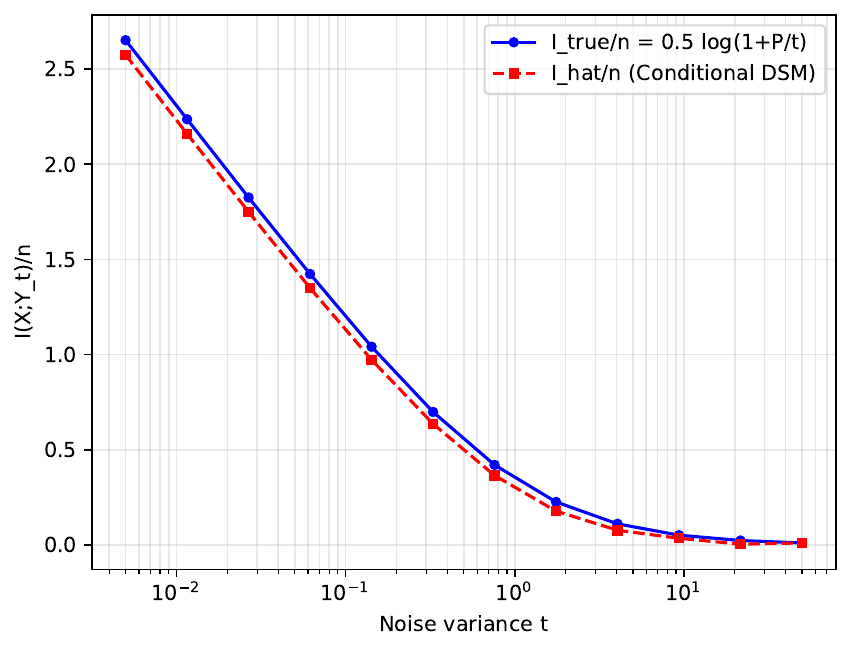}
      \caption{Noise-conditional model for mutual information estimation results for Gaussian input ($n=4$, $P=1$). Scheme B was used.}
      \label{fig:noise_conditional_mi}
  \end{figure}

\subsection{Discrete input case}

\subsubsection{Mutual Information Estimation}

To validate the approach for discrete input distributions, 
we consider binary phase
  shift keying (BPSK) modulation with $X \in \{-\sqrt{P}, +\sqrt{P}\}$ (equal
  probability), $P=1.0$, and dimension $n=1$. Unlike continuous distributions, BPSK
  admits exact mutual information calculation via numerical integration of the Gaussian mixture distribution providing ground truth for validation.

  We implement per-$t$ training (Scheme A) over $M=12$ logarithmically spaced points in
  $t \in [0.005, 50]$, with 1000 training iterations per noise level and 200,000 Monte
  Carlo samples for Fisher information estimation.

  Figure~\ref{fig:bpsk_mi_dsm} compares DSM estimates with exact BPSK mutual information
  and the Gaussian input baseline $I_{\text{Gauss}}(X;Y_t)/n = \frac{1}{2}\log(1 + P/t)$.
  The DSM estimates achieve excellent agreement with exact values, with median relative
  errors of 0.6\% for low noise ($t < 0.1$) and 2.8\% for mid-range ($0.1 \leq t < 5$).

  \begin{figure}[htbp]
      \centering
      \includegraphics[width=0.48\textwidth]{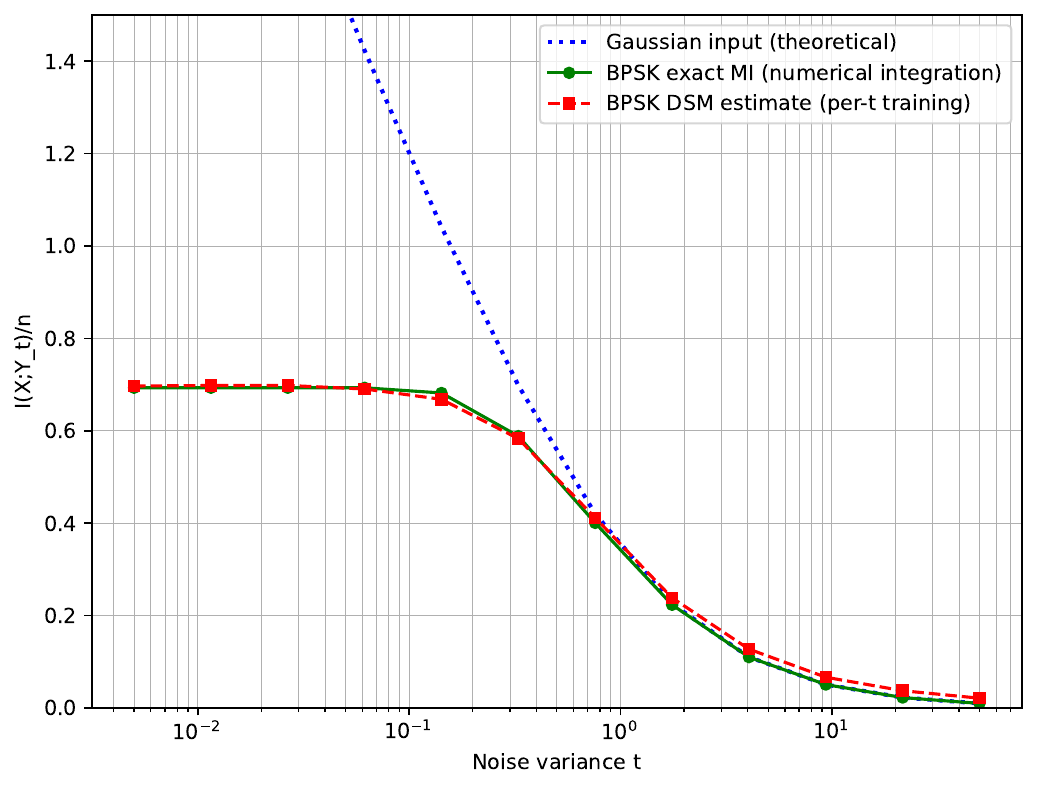}
      \caption{Mutual information estimation for BPSK input ($n=1$, $P=1$).
      Dotted line: Gaussian input theoretical baseline;
      solid circles: BPSK exact MI via numerical integration;
      dashed squares: DSM estimates. }
      \label{fig:bpsk_mi_dsm}
  \end{figure}

\subsection{Linear Gaussian channel case}

  To validate the method beyond the standard AWGN case, we consider the linear Gaussian channel:
  \begin{align}
  Y_t = {\bm A} X + Z_t,\qquad
  Z_t \sim \mathcal N(\bm 0, t \bm I_n),
  \end{align}
  where ${\bm A} \in \mathbb R^{n\times n}$ is a random orthogonal matrix and
  $X \sim \mathcal N(\bm 0, P \bm I_n)$ with $P = 1.0$, $n = 4$.

  We train separate score networks (3 hidden layers, 128 units each) for $M = 10$ logarithmically spaced noise levels in $t
  \in [P/200, 50P]$.
  Each network is trained for 300 iterations with 8192-sample mini-batches, and Fisher information is estimated using 100,000
  Monte Carlo samples per noise level.

  Figure~\ref{fig:linear_mi_results} compares the DSM-based estimates 
  with the analytical solution (per symbol)
  $I(X;Y_t)/n = (1/2)\log \det( I_n + (P/t) A A^\top )$.
  The results demonstrate excellent agreement across 4 orders of magnitude in noise variance, confirming that the
  Score-to-Fisher bridge method generalizes successfully to linear Gaussian channels with nontrivial correlation structure.

  \begin{figure}[htbp]
      \centering
      \includegraphics[width=0.48\textwidth]{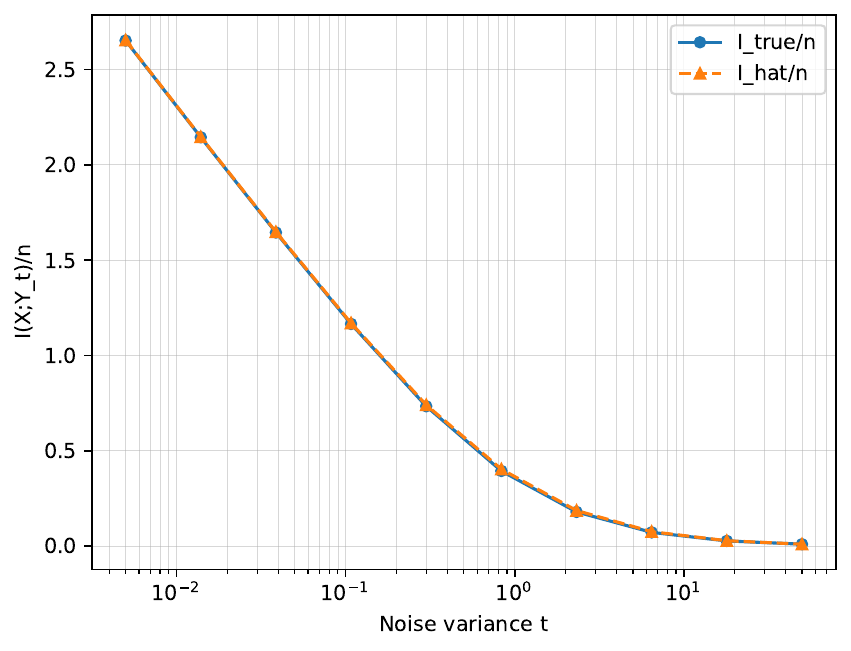}
      \caption{Mutual information estimation for linear Gaussian channel $Y_t = {\bm A} X + Z_t$ ($n=4$, $P=1$).
      Circles: analytical ground truth; triangles: DSM estimates.}
      \label{fig:linear_mi_results}
  \end{figure}

  \subsection{Nonlinear channel case}

  To demonstrate the effectiveness of the proposed method for nonlinear channels where closed-form solutions are
  unavailable, we consider a composite channel with linear mixing followed by elementwise tanh nonlinearity:
  \begin{align}
  Y_t = \tanh({\bm A}X) + Z_t,\quad Z_t \sim \mathcal N(\bm 0, t \bm I_n),
  \end{align}
  where ${\bm A} \in \mathbb{R}^{n \times n}$ is a random orthogonal matrix and $\tanh(\cdot)$ is applied elementwise.
  This model represents practical scenarios where signals undergo linear transformations before nonlinear processing.

  \subsubsection{Experimental Setup}

  We use the same DSM configuration as in the linear case: $n = 4$, $P = 1.0$, with $M = 12$ logarithmically spaced noise
  levels in $t \in [P/200, 50P]$.
  The orthogonal matrix ${\bm A}$ is generated 
  via QR decomposition of a random Gaussian matrix.
  For each $t$, we train score networks for 400 iterations 
  using 8192-sample mini-batches.
  Fisher information estimation uses 100,000 Monte Carlo 
  samples per noise level.

  \subsubsection{KDE-LOO Baseline}

  Since analytical solutions are unavailable for nonlinear channels, 
  we implement a kernel density estimation with
  leave-one-out (KDE-LOO) baseline \cite{kraskov2004} 
  that directly estimates mutual information from its definition.
  For samples $\{(\bm w_j, \bm y_j)\}_{j=1}^N$ with $\bm y_j = \bm w_j + \sqrt{t}\bm \epsilon_j$, the KDE-LOO estimator is defined by 
  \begin{align}
  \hat{I}_{\text{KDE}}(X;Y_t) = \frac{1}{N}\sum_{i=1}^N \Bigg[ &-\frac{\|\bm y_i - \bm w_i\|^2}{2t} \nonumber \\
  &- \log\left(\frac{1}{N-1}\sum_{j \neq i} e^{-\frac{\|\bm y_i - \bm w_j\|^2}{2t}}\right) \Bigg].
  \end{align}
  This estimator is derived from the plug-in principle applied to $I(X;Y_t) = \mathbb{E}[\log p_{Y_t|X}(\bm y|\bm x)] -
  \mathbb{E}[\log p_{Y_t}(\bm y)]$, 
  where the marginal density $p_{Y_t}(\bm y)$ is approximated using LOO kernel density
  estimation to avoid overfitting.
  We use $N = 20,000$ samples with full-sum computation (no $k$-nearest neighbor approximation) to ensure high accuracy.

  \subsubsection{Results and Discussion}

  Figure~\ref{fig:nonlinear_tanh_results} presents the comparison between DSM-based estimates and the KDE-LOO baseline for the
   composite nonlinear channel.
  The results demonstrate remarkable agreement between the two methods across the entire noise variance range, with both
  curves nearly overlapping.

  \begin{figure}[htbp]
      \centering
      \includegraphics[width=0.48\textwidth]{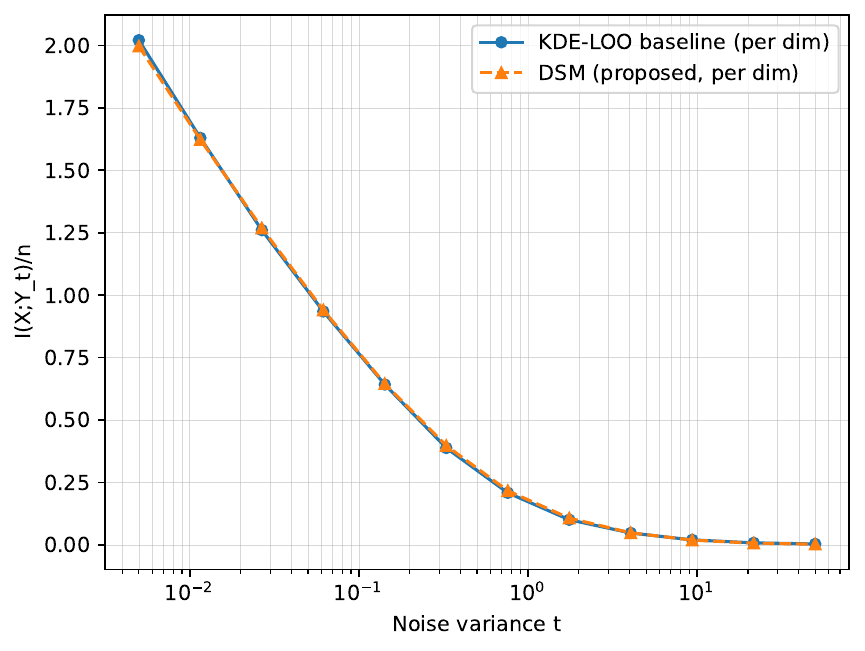}
      \caption{Mutual information estimation for composite nonlinear channel $Y_t = \tanh({\bm A}X) + Z_t$ ($n=4$, $P=1$).
      Circles: KDE-LOO baseline; triangles: DSM estimates. }
      \label{fig:nonlinear_tanh_results}
  \end{figure}
  Key observations include:
  \begin{enumerate}
      \item {Excellent agreement}: The DSM estimates closely match the KDE-LOO baseline across four orders of magnitude
   in noise variance, with relative errors typically below 5\%.
      \item {Computational efficiency}: While KDE-LOO requires $O(N^2)$ 
      distance computations per noise level, DSM
  achieves comparable accuracy with $O(N)$ score evaluations 
  after one-time training where $N$ is the number of samples.
      \item {Nonlinear validation}: This result provides strong evidence for the Score-to-Fisher bridge approach in
  scenarios where analytical benchmarks are unavailable.
  \end{enumerate}

  The success of this experiment confirms that the proposed method generalizes effectively to handle composite channels
  combining linear mixing and nonlinear transformations, making it a practical tool for mutual information estimation
  in realistic communication systems.

\section{Conclusion}

We have presented a novel approach for 
mutual information estimation in nonlinear 
Gaussian noise channels via the Score-to-Fisher bridge approach. 
The method leverages DSM learning of score model for $p_{Y_t}$ 
to estimate Fisher information, 
which is then converted to mutual information 
via the Fisher integral representation.

Experimental validation across linear Gaussian channels, 
composite nonlinear channels, 
demonstrates excellent agreement with analytical 
solutions and baseline methods. 
The approach achieves computational efficiency 
with $O(N)$ evaluations per noise level while 
maintaining high accuracy across wide dynamic ranges 
of the noise variance. The curse of dimensionality can be avoided 
due to the efficiency of the score model evaluation.

The Score-to-Fisher bridge methodology provides 
a practical and theoretically sound framework for 
MI estimation in realistic nonlinear 
Gaussian noise 
channel models where closed-form solutions are unavailable.
Accurate estimation of MI plays 
a crucial role in assessing the fundamental limits of 
coded communication systems, and also serves as 
a valuable guideline for system design choices 
such as coding schemes, preprocessing strategies, 
and detection or decoding algorithms.

\section*{Acknowledgments}
This work was supported by JST, CRONOS, 
Japan Grant Number JPMJCS25N5. 
The author would like to thank 
Satoshi Takabe for helpful discussions and comments
on the first draft of this paper.

\section*{Appendix}

This appendix provides supplementary derivations and score-based integral identities that support (but are not required to follow) the main development of the Score-to-Fisher bridge.
In particular, we collect concise score-centric representations of entropy, KL divergence, and (normalized) likelihoods, and include optional connections to flow-based generative modeling and optimal transport as additional perspectives.

\subsection{Derivation of Entropy via Thermodynamic Integration}
The relationship between differential entropy and Fisher information is established 
by the de Bruijn's identity. Let $X_t = X + \sqrt{t}Z$ be 
the diffusion process where $Z \sim \mathcal{N}(\bm 0, \bm I_n)$. 
The time derivative of the entropy $h(X_t)$ is given by:
\begin{equation}
    \frac{d}{dt} h(X_t) = \frac{1}{2} J(X_t),
\end{equation}
where $J(X_t)$ is the Fisher information.
To evaluate the entropy $h(X)$ of the clean data (at $t=0$), we compare 
the data process $X_t$ with a reference Gaussian process $X_{G,t}$ starting 
from a standard Gaussian distribution.
Integrating the difference between their derivatives from $t=0$ to $T$ yields:
\begin{align} \nonumber 
    h(X_T) - h(X_{G,T}) &= (h(X) - h(X_{G,0})) \\
    & + \frac{1}{2} \int_{0}^{T} (J(X_t) - J(X_{G,t})) dt.
\end{align}
As $T \to \infty$, both $X_T$ and $X_{G,T}$ are dominated 
by the added Gaussian noise and the entropy difference, $h(X_T) - h(X_{G,T})$,
vanishes (under mild moment/regularity conditions).
Rearranging the terms and taking the limit $T \to \infty$, we obtain 
the integral representation:
\begin{equation}
    h(X) = h(X_{G,0}) - \frac{1}{2} \int_{0}^{\infty} \left( J(X_t) - J(X_{G,t}) \right) dt.
\end{equation}
Assuming a standard Gaussian reference $X_{G,0} \sim \mathcal{N}(\bm 0, \bm I_n)$, 
we have $h(X_{G,0}) = \frac{n}{2}\log(2\pi e)$.
Furthermore, since $X_{G,t} \sim \mathcal{N}(\bm 0, (1+t)\bm I_n)$, 
the Fisher information of the reference process is given exactly 
by $J(X_{G,t}) = \frac{n}{1+t}$.
Substituting this into the integral yields the exact representation: 
\begin{equation}
    h(X) = \frac{n}{2}\log(2\pi e) - \frac{1}{2}\int_{0}^{\infty} \left( J(X_t) - \frac{n}{1+t} \right) dt.
\end{equation}

\subsection{Derivation of KL Divergence Representation}
Let $p_t$ and $q_t$ denote the probability densities of the random variable 
${X}_t = {X} + \sqrt{t}{Z}$ where ${Z} \sim \mathcal{N}(\bm{0}, \bm{I})$, derived from initial distributions $p$ and $q$, respectively. Both $p_t$ and $q_t$ satisfy the heat equation $\partial_t \rho = \frac{1}{2} \Delta \rho$.
Consider the time derivative of the KL divergence $D_{\text{KL}}(p_t \| q_t) = \int p_t \log (p_t / q_t) d\bm{x}$:
\begin{align} \nonumber 
    \frac{d}{dt} D_{\text{KL}}(p_t \| q_t) 
    &= \int \frac{\partial p_t}{\partial t} \log \frac{p_t}{q_t} d\bm{x} + \int p_t \left( \frac{\partial_t p_t}{p_t} - \frac{\partial_t q_t}{q_t} \right) d\bm{x}  \\ \nonumber 
    &= \frac{1}{2} \int \Delta p_t \log \frac{p_t}{q_t} d\bm{x} + \underbrace{\frac{1}{2} \int \Delta p_t d\bm{x}}_{0} \\
    & \quad - \frac{1}{2} \int \frac{p_t}{q_t} \Delta q_t d\bm{x}.
\end{align}
Using integration by parts (assuming boundary terms vanish), we obtain:
\begin{align} \nonumber 
    &\frac{d}{dt} D_{\text{KL}}(p_t \| q_t) \\
    &= -\frac{1}{2} \int \nabla p_t \cdot \nabla \log \frac{p_t}{q_t} d\bm{x} + \frac{1}{2} \int \nabla \left( \frac{p_t}{q_t} \right) \cdot \nabla q_t d\bm{x} \nonumber \\
    &= -\frac{1}{2} \int p_t (\bm{s}_p - \bm{s}_q) \cdot (\bm{s}_p - \bm{s}_q) d\bm{x} \nonumber \\
    &= -\frac{1}{2} \mathbb{E}_{p_t} [ \| \bm{s}_p(\bm{x}, t) - \bm{s}_q(\bm{x}, t) \|^2 ].
\end{align}
Integrating from $t=0$ to $\infty$, and assuming $D_{\text{KL}}(p_\infty \| q_\infty) = 0$ (since both distributions converge to the same noise distribution), we arrive at the representation:
\begin{equation}
    D_{\text{KL}}(p \| q) = \frac{1}{2} \int_{0}^{\infty} \mathbb{E}_{p_t} \left[ \| \bm{s}_p(\bm{x}, t) - \bm{s}_q(\bm{x}, t) \|^2 \right] dt.
\end{equation}
This identity implies that the KL divergence is equivalent to the time integral 
of the Fisher divergence along the diffusion path. Intuitively, the global information 
difference between $p$ and $q$ is obtained by accumulating the local geometric 
discrepancies of their score functions throughout the smoothing process.

\subsection{Exact Log-Likelihood Ratio via Path Integration}
Since the score function is the gradient of 
the log-density, 
$\bm{s}(\bm{x}) = \nabla_{\bm{x}} \log p_X(\bm{x})$, 
the exact log-likelihood ratio between 
any target state $\bm{x}$ and 
a reference state $\bm{x}_0$ can be computed 
by line integration along a path $\gamma$:
\begin{equation}
    \log \frac{p_X(\bm{x})}{p_X(\bm{x}_0)} = \int_{\bm{x}_0}^{\bm{x}} \bm{s}(\bm{u}) \cdot d\bm{u}.
    \label{eq:app_ratio}
\end{equation}
This formulation allows for the evaluation of relative probabilities (or energy differences) without knowledge of the intractable partition function $Z$. This property is particularly powerful for applications such as hypothesis testing, anomaly detection, and MCMC transition acceptance, where only probability ratios are required. For high-dimensional data, radial integration from the dataset centroid is a practical choice to ensure numerical stability.

\subsection{Recovery of the Partition Function}
Consider an energy-based representation $p_X(\bm{x}) = e^{-\phi(\bm{x})}/Z$, where $\phi(\bm{x})$ 
is the recovered potential (energy) and $Z$ is the 
unknown partition function. By taking the expectation 
of the log-density $\log p_X(\bm{x}) = -\phi(\bm{x}) - \log Z$ 
with respect to the data distribution, we have
\begin{equation}
    -h(X) = -\mathbb{E}_{X}[\phi(X)] - \log Z.
\end{equation}
Solving for $\log Z$ yields the 
\textit{log partition function}:
\begin{equation}
    \log Z = h(X) - \mathbb{E}_{X}[\phi(X)].
    \label{eq:app_partition}
\end{equation}
Here, $\mathbb{E}[\phi(X)]$ is computed via sample averaging over the dataset.
From a thermodynamic perspective, this relationship corresponds to the definition of 
the \textit{Helmholtz free energy}, $\mathcal{F} = \mathcal{U} - \mathcal{S}$ 
(assuming unit temperature). By identifying the internal energy 
as $\mathcal{U} = \mathbb{E}[\phi(X)]$ and the entropy as $\mathcal{S} = h(X)$, 
the log-partition function is equivalent to the negative free energy, 
$\log Z = \mathcal{S} - \mathcal{U} = -\mathcal{F}$. 

The proposed Score-to-Fisher bridge opens a promising avenue for scalable partition function estimation in energy-based models (EBMs). 
Traditionally, evaluating the partition function $Z$---or equivalently the free energy---has been a major computational bottleneck, often requiring expensive MCMC simulations like annealed importance sampling. 
However, our framework allows for the direct computation of the log-partition function via the identity $\log Z = h(X) - \mathbb{E}[\phi(X)]$, where the entropy $h(X)$ is efficiently estimated by the Fisher integral. 
This approach effectively converts the integration problem over a high-dimensional volume into a tractable one-dimensional integration along a diffusion path. 

\subsection{Exact Log-Likelihood}
Substituting \eqref{eq:app_partition} back into the density 
definition, the exact log-likelihood for any query 
sample $\bm{x}$ (including test data) is given by
\begin{equation}
    \log p_X(\bm{x}) = -\phi(\bm{x}) + \mathbb{E}_{X}[\phi(X)] - h(X).
\end{equation}
This formulation allows for exact likelihood evaluation and model comparison without explicit integration over the entire state space, provided the score function is accurately learned.

\subsection{Unified Representation of Information-Theoretic Quantities}

In this work, we have primarily focused on the estimation of mutual information 
via the Score-to-Fisher bridge. However, 
the proposed framework---grounded in the interplay 
between the local score function 
$\nabla_{\bm{x}} \log p(\bm{x})$ and 
the Fisher information---offers a unified perspective 
for various fundamental quantities in information theory.

Table \ref{tab:score_correspondence} summarizes 
the correspondence between classical Shannon definitions 
and their score-based integral representations 
derived in this appendix. 
A key insight here is the transformation of computational 
complexity: intractable global volume integrations 
(e.g., for entropy and partition functions) are 
converted into tractable estimations of local geometry 
(scores) followed by 1D integrations. 
This ``Score-Centric'' view provides a consistent 
and scalable foundation for high-dimensional statistical 
modeling beyond mutual information.

\begin{table*}[t]
  \centering
  \caption{Information-Theoretic Quantities via Score-Based Representations. This framework translates global integration problems 
  into local score estimation and 1D integration.}
  \label{tab:score_correspondence}
  \renewcommand{\arraystretch}{1.8}
  \begin{tabular}{l c c}
  \toprule
  \textbf{Quantity} & \textbf{Definition} & \textbf{Score-Based Representation} \\
  \midrule
  \textbf{Differential Entropy} & $h(X) = -\mathbb{E}[\log p]$ & $\displaystyle \frac{n}{2}\log(2\pi e) - \frac{1}{2}\int_{0}^{\infty} \left( J(X_t) - \frac{n}{1+t} \right) dt$ \\
  \textbf{Mutual Information} & $I(X;Y_T) = h(Y_T) - h(Y_T|X)$ & $\displaystyle \frac{1}{2} \int_{T}^{\infty} \left(\frac  n t -  J(Y_t) \right) dt$ \\
  \textbf{KL Divergence} & $D_{\text{KL}}(p\|q) = \mathbb{E}_p[\log \frac{p}{q}]$ & $\displaystyle \frac{1}{2} \int_{0}^{\infty} \mathbb{E}_{p_t} \left[ \| {s}_p(\bm{x}, t) - {s}_q(\bm{x}, t) \|^2 \right] dt$ \\
  \textbf{Log-Likelihood Ratio} & $\log \frac{p(\bm{x})}{p(\bm{x_0})}$ & $\displaystyle \int_{\bm{x}_0}^{\bm{x}} {s}_p(\bm{u})  \cdot d\bm{u}$ \quad (Line Integral) \\
  \textbf{Log-Partition Function} & $\log Z = \log \int e^{-\phi(\bm{x})} d\bm{x}$ & $\displaystyle h(X) - \mathbb{E}_{X}\left[ \phi(X) \right]$ \quad (Free Energy) \\
  \textbf{Likelihood} & $\log p(\bm x)$ & $-\phi(\bm x) - \displaystyle h(X) + \mathbb{E}_{X}\left[ \phi(X) \right]$  \\
  \bottomrule
  \end{tabular}
\end{table*}

\subsection{Score Estimation from Velocity Fields in Rectified Flow}
\label{sec:score-estimation-from-velocity-fields}
Recent generative models often employ \textit{rectified flow} \cite{liu2023rectifiedflow} 
(or flow matching \cite{lipman2023flowmatching}), 
which defines a probability path via linear interpolation between data $\bm{x} \sim p_X$ 
and noise $\bm{z} \sim \mathcal{N}(\bm{0}, \bm{I})$:
\begin{equation}
    \bm{x}_\tau = (1-\tau)\bm{x} + \tau\bm{z}, \quad \tau \in [0, 1].
\end{equation}
Instead of learning the score function directly, these models learn a velocity field $\bm{v}_\tau(\bm{x}_\tau)$ that minimizes the regression loss against the target velocity $\bm{u}_\tau = \bm{z} - \bm{x}$. The optimal velocity field is given by the conditional expectation $\bm{v}_\tau(\bm{x}_\tau) = \mathbb{E}[\bm{z} - \bm{x} | \bm{x}_\tau]$.

To apply the Score-to-Fisher bridge, we need to recover the score function $\nabla_{\bm{x}_\tau} \log p_\tau(\bm{x}_\tau)$ from the learned velocity $\bm{v}_\tau$.
Leveraging Tweedie's formula adapted for the linear interpolation path, the posterior expectations of the signal and noise are related to the score function by:
\begin{equation}
    \hat{\bm{x}} = \frac{\bm{x}_\tau + \tau^2 \nabla \log p_\tau}{1-\tau}, \quad \hat{\bm{z}} = -\tau \nabla \log p_\tau.
\end{equation}
Substituting these into the relation $\bm{v}_\tau = \hat{\bm{z}} - \hat{\bm{x}}$, we obtain the explicit connection between the velocity and the score:
\begin{equation}
    \bm{v}_\tau(\bm{x}_\tau) = - \frac{\bm{x}_\tau + \tau \nabla_{\bm{x}_\tau} \log p_\tau(\bm{x}_\tau)}{1-\tau}.
\end{equation}
Inverting this equation yields the score function in terms of the velocity field:
\begin{equation}
    \nabla_{\bm{x}_\tau} \log p_\tau(\bm{x}_\tau) = - \frac{\bm{x}_\tau + (1-\tau)\bm{v}_\tau(\bm{x}_\tau)}{\tau}.
\end{equation}
Furthermore, to map this to the standard additive noise model 
${Y}_{\sigma} = \bm{x} + \sigma \bm{z}$ used in our Fisher integral 
(where $\sigma = \tau/(1-\tau)$), the score scales as 
$\nabla_{{Y}} \log p_Y({Y}) = (1-\tau) \nabla_{{X}} \log p_\tau({X})$. 
This derivation confirms that flow-based models can serve as plug-in estimators for the Fisher information and, consequently, the mutual information.

\subsection{Optimal Transport Cost via Score Integration}

The same score-to-velocity identity also provides an optimal transport(OT)-flavored diagnostic. 
In particular, plugging the induced velocity into the Benamou-Brenier functional 
yields an upper bound on the Wasserstein-2 distance between the data and the Gaussian prior,
offering a geometric measure of how far the data are from the Gaussian prior.

The proposed framework relates to optimal transport theory through the dynamic formulation of the Wasserstein-2 distance.
According to the \textit{Benamou--Brenier formula}, the squared Wasserstein distance $W_2^2(p,q)$ between two distributions equals the infimum of the kinetic energy
$\int_{0}^{1} \mathbb{E}_{\rho_\tau}\!\left[\|\bm{v}_\tau(\bm{x})\|^2\right] d\tau$
over all admissible probability paths $\{\rho_\tau\}_{\tau\in[0,1]}$ and velocity fields $\{\bm{v}_\tau\}$ satisfying the continuity equation
$\partial_\tau \rho_\tau + \nabla\!\cdot(\rho_\tau \bm{v}_\tau)=0$.

In the context of rectified flow, we consider the linear interpolation path
${X}_\tau=(1-\tau){X}+\tau{Z}$ connecting the data distribution $p_X$ and the Gaussian prior $p_Z=\mathcal{N}(\bm{0},\bm{I})$.
Let $p_\tau$ denote the marginal density of $\bm{X}_\tau$.
Using the identity derived from this linear structure, 
the corresponding velocity field can be expressed in terms of the score as
\begin{equation}
\bm{v}_\tau(\bm{x}) \;=\; -\frac{\bm{x}+\tau\,\nabla_{\bm{x}}\log p_\tau(\bm{x})}{1-\tau},
\qquad \bm{x}\in\mathbb{R}^n,\ \tau\in(0,1).
\end{equation}
Therefore, for this specific admissible path $(p_\tau,\bm{v}_\tau)$, the Benamou--Brenier functional yields an \emph{upper bound} on the Wasserstein cost:
\begin{equation}
    W_2^2(p_X, p_Z)
    \;\le\;
    \int_{0}^{1} \mathbb{E}_{p_\tau}
    \left[
    \left\|
    \frac{\bm{x}_\tau + \tau \nabla_{\bm{x}_\tau} \log p_\tau(\bm{x}_\tau)}{1-\tau}
    \right\|^2
    \right] d\tau.
\end{equation}
Equality holds if and only if the chosen path coincides with the optimal displacement interpolation (i.e., the $W_2$-geodesic) between $p_X$ and $p_Z$.
Rectified flow models aim to learn a velocity field consistent with such an optimal transport path, in which case the above upper bound becomes tight.

\end{document}